\documentclass[11pt]{amsart}
\usepackage{amsthm,amsmath,amssymb,mathtools}
\usepackage{tikz}
\usepackage[utf8]{inputenc}
\usepackage[hidelinks]{hyperref}
\usepackage{pgfplots}
\usepackage[final]{microtype}
\pgfplotsset{compat=1.18}

\usepackage[margin=1in]{geometry}
\usepackage[sortcites,backend=biber,hyperref=true,maxbibnames=99,style=numeric]{biblatex}
\newcommand{\N}{\mathbb{N}}
\newcommand{\R}{\mathbb{R}}
\newcommand{\D}{\mathcal{D}}
\newcommand{\Sphere}{\mathcal{S}}
\newcommand{\Z}{\mathbb{Z}}
\newcommand{\Q}{\mathbb{Q}}

\DeclareMathOperator{\dimH}{dim_H}

\newtheorem{theorem}{Theorem}
\newtheorem{lemma}[theorem]{Lemma}
\newtheorem{observation}[theorem]{Observation}

\theoremstyle{remark}
\newtheorem*{claim}{Claim}

\title{Algorithmic Information Bounds for Distances and Orthogonal Projections}

\author[P. Cholak]{Peter Cholak}
\address{University of Notre Dame}  
\email{Peter.Cholak.1@nd.edu}

\author[M. Cs\"ornyei]{Marianna Cs\"ornyei}
\address{University of Chicago}  
\email{csornyei@math.uchicago.edu}

\author[N. Lutz]{Neil Lutz}
\address{University of Notre Dame}
\email{nlutz2@nd.edu}

\author[P. Lutz]{Patrick Lutz}
\address{University of Michigan}
\email{pglutz@umich.edu}

\author[E. Mayordomo]{Elvira Mayordomo}
\address{Universidad de Zaragoza}
\email{elvira@unizar.es}

\author[D. Stull]{Donald M. Stull}
\address{University of Chicago}
\email{dmstull@uchicago.edu}

\begin{document}

\begin{abstract}
    We introduce a new technique for proving bounds on the Kolmogorov complexity of geometric objects in Euclidean space, such as points and lines. We apply this technique to prove two theorems on algorithmic information theory, both of which have consequences for well-known problems in geometric measure theory. First, we show that for any point $x$ in the plane and any other point $y$ sufficiently independent of $x$, the distance between $x$ and $y$ retains at least half the complexity of the original point $x$. By the point-to-set principle of J. Lutz and N. Lutz, this yields an improved lower bound on the Hausdorff dimension of pinned distance sets, a topic closely related to Falconer's distance set conjecture. Second, we prove an analogous result for orthogonal projections: for any point $x$ in the plane and any line through the origin which is sufficiently independent of $x$, the projection of $x$ onto that line retains at least half the complexity of $x$. As a consequence, we obtain a generalization of a theorem of Bourgain on exceptional sets for orthogonal projections.
\end{abstract}

\maketitle

\section{Introduction}

This paper is a contribution to a growing area of research which connects algorithmic information theory and geometric measure theory and which is based on a close correspondence between Kolmogorov complexity and Hausdorff dimension. This correspondence has two facets. First, there is a conceptual analogy between Kolmogorov complexity and Hausdorff dimension. For example, using Kolmogorov complexity, one can define a notion of effective dimension for objects such as infinite binary sequences and real numbers which shares many features with Hausdorff dimension~\cite{Stai93, CaiHart94, Lutz03b}. Second, there is a specific technical result, proved by J. Lutz and N. Lutz and known as the \emph{point-to-set principle}, which allows one to use facts about Kolmogorov complexity and effective dimension to prove theorems about Hausdorff dimension~\cite{LutLut18}.

In recent years, the point-to-set principle has been applied to make progress on a few long-standing questions in geometric measure theory. For example, N. Lutz and Stull used it to improve the best then-known lower bound for the Furstenberg set conjecture (a generalization of the Kakeya conjecture)~\cite{LutStu20} and Stull used it to improve the best known lower bound for Falconer's distance set conjecture~\cite{Stull22c}, an important open question in geometric measure theory~\cite{Iosevich2019}.

The main goal of this paper is to introduce a new technique for proving lower bounds on the Kolmogorov complexity of geometric objects in Euclidean space. One of our primary motivations for introducing this new technique is that, when used in conjunction with the point-to-set principle, it seems especially useful for proving results in geometric measure theory which had previously appeared difficult to approach using the methods of algorithmic information theory. However, even ignoring these applications to geometric measure theory, we believe that the theorems in algorithmic information theory which can be proved using our technique are interesting in their own right.

In this paper, we will give two example applications of our technique. In each case, we will first use our technique to prove a theorem on algorithmic information theory and then use the point-to-set principle to transfer this theorem to the setting of geometric measure theory. We have included these applications not just for their own sake, but also as a demonstration of the utility of our technique, which we believe likely has further applications within geometric measure theory. In fact, in unpublished work, Fiedler~\cite{Fiedler25} has already used our technique to prove a new theorem on the Hausdorff dimension of planar extensions of subsets of Euclidean space, a topic which is closely connected to the Kakeya Conjecture.

Our first application concerns the following question: suppose that $x$ and $y$ are points in $\R^2$ and that $x$ is relatively complex (in the sense of Kolmogorov complexity). Does this imply anything about the complexity of $|x - y|$, the distance between $x$ and $y$? Of course, if we make no assumptions whatsoever about $y$ then the answer is no. However, we will show that if $x$ and $y$ are reasonably independent then $|x - y|$ must retain at least half the complexity of $x$.

Slightly more precisely, let $K_r(x)$ denote the Kolmogorov complexity of an $r$-bit approximation of $x$ (i.e. of the first $r$ bits of both coordinates of $x$) and let $K_r^y(z)$ denote the Kolmogorov complexity of an $r$-bit approximation of $z$ relative to an oracle for $y$. We will show, in Theorem~\ref{thm:pinned}, that if $x$ and $y$ satisfy a mild independence condition then
\begin{equation}\label{eq:pinnedintro}
  K^y_r(|x - y|) \geq \frac{K_r(x)}{2} - o(r).
\end{equation}

We will then use this theorem to prove a new lower bound on the Hausdorff dimension of pinned distance sets. Given a set $E \subseteq \R^2$, the \emph{distance set} of $E$ is the set $\Delta E = \{|x - y| \mid x, y \in E\}$ of pairwise distances between elements of $E$. Given a set $E \subseteq \R^2$ and a point $x \in \R^2$, the \emph{pinned distance set} is the set $\Delta_xE = \{|x - y| \mid y \in E\}$ of distances between $x$ and elements of $E$. Falconer's distance set conjecture states that if $\dim_H(E) > 1$ then $\dim_H(\Delta E) = 1$. Most progress on this conjecture has come from estimating the dimension of pinned distance sets. Using our Theorem~\ref{thm:pinned}, we establish in Theorem~\ref{thm:pinnedclassical} that for every analytic set $E \subseteq \R^2$ satisfying $s := \dim_H(E) \leq 1$,
\[
  \sup_{x \in E}\dim_H(\Delta_xE) \geq \frac{3s}{4}.
\]
This improves on the best previously known bounds, due to Shmerkin and Wang~\cite{ShmWang25} and Fiedler and Stull~\cite{FieStu24a}, for all $0 < s < 1$.

Our second application concerns the complexity of orthogonal projections of points in the plane. In particular, for any point $x \in \R^2$ and direction $e \in S^1$, let $p_ex$ denote the orthogonal projection of $x$ onto the line in direction $e$. In Theorem~\ref{thm:orth}, we show that if $e$ is sufficiently independent of $x$ then
\begin{equation}\label{eq:orthintro}
K^e_r(p_ex) \geq \frac{K_r(x)}{2} - o(r).
\end{equation}
In other words, the projection of $x$ retains at least half the complexity of $x$.

We then apply this to generalize a theorem of Bourgain~\cite{MR2763000} on the Hausdorff dimension of exceptional sets for orthogonal projections. Our theorem extends Bourgain's from analytic sets to a more general class of sets which admit \emph{optimal Hausdorff oracles}~\cite{Stull22a}. This result is similar in spirit to Shmerkin's~\cite{shmerkin2023non} generalization of Bourgain's theorem in the context of distance sets.

We now outline the main idea of our technique. We will illustrate the idea in the context of our Theorem~\ref{thm:pinned}, discussed above. Suppose that the conclusion of that theorem fails and that in general, $K_r^y(|x - y|)$ may be much less than $K_r(x)/2$, even when $x$ and $y$ are fairly independent. By standard reasoning, we may assume that it actually fails for two different values of $y$, say $y_1$ and $y_2$ (with respect to the same value of $x$).

It is now tempting to reason as follows: if we know $y_1$ and $y_2$, as well as the distances $|x - y_1|$ and $|x - y_2|$ then we can easily recover $x$ with very little additional information---the main point is that the circles with centers $y_1$ and $y_2$ and radii $|x - y_1|$ and $|x - y_2|$ should in general intersect in only two points. Thus the total information content of $x$ should be no larger than the combined information content of $y_1, y_2, |x - y_1|,$ and $|x - y_2|$. Furthermore, since we are assuming that $y_1$ and $y_2$ are independent of $x$, we can ignore the contribution of $y_1$ and $y_2$ and conclude that the information content of $x$ should be no larger than the combined information content of $|x - y_1|$ and $|x - y_2|$. Thus it would seem that at least one of $|x - y_1|$ and $|x - y_2|$ must have at least half the information content of $x$, which is exactly what we want to conclude.

However, there is an important problem with this argument. Whenever we talk about the information content of geometric objects, like points in $\R^2$, what we really mean is the Kolmogorov complexity of finite approximations of these objects. When we assumed that the theorem fails for $x, y_1$ and for $x, y_2$, what we actually meant was that for infinitely many $r$, we have $K^{y_1}_r(|x - y_1|) < K_r(x)/2$ and for infinitely many $r$, $K^{y_2}_r(|x - y_2|) < K_r(x)/2$. But the values of $r$ for which these inequalities hold may be different for $y_1$ and $y_2$. In other words, it may be that at every level of precision, either the finite approximation to $|x - y_1|$ or the finite approximation to $|x - y_2|$ at that precision has low complexity, but they happen to never both have low complexity at the same time.

What we really need in order for the argument sketched above to work is not just some $y_1$ and $y_2$ for which the theorem fails with respect to the same $x$, but $y_1$ and $y_2$ for which the theorem fails with respect to the same $x$ \emph{and the same level of precision $r$}. At the core of our technique is a method for finding such $y_1, y_2$ and $r$.

The key idea is that we replace $x$ with a different point with similar properties, which we will refer to as a \emph{surrogate point} for $x$. In order to do so, we first show that if the theorem fails then for each $r$, there are many pairs $x, y$ for which it fails. We then use a combinatorial lemma (Lemma~\ref{lem:comb}, proved by a Cauchy-Schwarz double counting argument) to show that it fails for so many pairs that we can find two pairs $x_1, y_1$ and $x_2, y_2$ with $x_1 = x_2$. Of course, this is only a rough description of the idea; the full details will be given in the proofs of Theorems~\ref{thm:pinned} and~\ref{thm:orth} below.

This work is an instance of applying \emph{effective methods} to geometric measure theory via the point-to-set principle~\cite{LutLut18}, an approach that has had significant recent success in studying the fractal dimensions of orthogonal projections~\cite{LutStu24,Stull22a,FieStu26}, pinned distance sets~\cite{Stull22c,FieStu23,FieStu24a,AlBuWi24}, and other geometric sets~\cite{LutStu20,Lutz21,NaPuS24,Shen23,LuQiYu24,BusFie24}. Beyond the measure-theoretic applications, pursuing these approaches has also sparked new innovations within the theory of computing, including results on the algorithmic information properties of points on specific lines~\cite{LutLut15,LutStu20,LutStu22,Stull25,LuMaWh26} and point-to-set principles for other domains and other types of fractal dimension~\cite{LuLuMa23a,LuLuMa23b,Mayordomo25,MayNie25}.

\section{Preliminaries}

\subsection{Kolmogorov Complexity in Discrete Domains}\label{ssec:strings}

We now briefly state the basic definitions and essential properties of (prefix) Kolmogorov complexity, which will be used repeatedly in our proofs. For a detailed treatment of this topic, see Chapter 3 of~\cite{LiVit19}.

Let $U$ be an oracle universal prefix Turing machine, which will remain fixed throughout this paper. For all binary strings $\sigma,\tau\in\{0,1\}^*$, the \emph{conditional (prefix) Kolmogorov complexity} of $\sigma$ given $\tau$ is
\[K(\sigma\mid\tau)=\min\big\{|\pi|\;\big|\; \pi\in\{0,1\}^*\text{ and }U^\emptyset(\pi,\tau)=\sigma\big\},\]
and the \emph{(prefix) Kolmogorov complexity} of $\sigma$ is
\[K(\sigma)=K(\sigma\mid\lambda),\]
where $\lambda$ denotes the empty binary string. Standard binary encoding readily extends these definitions to other countable domains, including natural numbers, rationals, dyadic rationals, and vectors of these.

Prefix Kolmogorov complexity obeys the \emph{symmetry of information}: For all $\sigma,\tau\in\{0,1\}^*$,
\begin{equation}\label{eq:soi}
    K(\sigma,\tau)=K(\sigma\mid \tau,K(\tau))+K(\tau)+O(1).
\end{equation}

The prefix Kolmogorov complexity of a string is never significantly longer than the string's length. In particular, for all $n\in\N$,
\begin{equation}\label{eq:nat}
    K(n)\leq \log n+2\log\log n+O(1).
\end{equation}

A family $\{E_r\mid r\in\N\}$ of sets is \emph{uniformly computably enumerable} if the set $\{\langle x,r\rangle\mid x\in E_r\}$ is computably enumerable. Where the dependence on $r$ is clear, we will say that an individual set $E_r$ in such a family as uniformly computably enumerable. We will often show that a set is large by showing that it is uniformly computably enumerable and contains objects of high complexity. In particular, if $\sigma,\tau\in\{0,1\}^*$ and $E$ is a finite set such that $\sigma\in E$ and $E$ is uniformly computably enumerable given $\tau$, then
\begin{equation}\label{eq:cardcomp}
    K(\sigma\mid\tau)\leq\log|E|+2\log\log|E|+O(1).
\end{equation}
Intuitively, this is true because, given $\tau$, we could algorithmically specify $\sigma\in E$ by giving its position in the enumeration of $E$, and this position is at most $|E|$.

\subsection{Kolmogorov Complexity in Euclidean Spaces}

To apply Kolmogorov complexity in Euclidean spaces, we consider the complexity of finite-precision rational approximations. 

For $x\in\R^n$ and $\varepsilon>0$, let $\mathcal{B}_\varepsilon(x)$ denote the open ball of radius $\varepsilon$ about $x$. For all $\sigma\in\{0,1\}^*$, $E,F\subseteq\R^n$, $r,s\in\N$, and $x,y\in\R^n$, define
\begin{align*}
    K(E\mid\sigma)&=\min_{q\in E\cap\Q^n}K(q\mid \sigma)\\
    K(E)&=K(E\mid\lambda)\\
    K(E\mid F)&=\max_{q\in F\cap\Q^n}K(E\mid q)\\
    K(\sigma\mid E)&=\max_{q\in E\cap\Q^n}K(\sigma\mid q)\\
    K_r(x)&=K(\mathcal{B}_{2^{-r}}(x))\\
    K_{r,s}(x\mid y)&=K(\mathcal{B}_{2^{-r}}(x)\mid \mathcal{B}_{2^{-s}}(y)).
\end{align*}

While $K_r(x)$ is defined as the minimum complexity of a rational point that is $2^{-r}$-close to $x$, it is often more convenient to analyze the complexity of a specific rational point that is close to $x$, namely, the point obtained by truncating, $r$ bits to the right of the binary point, the binary expansions of each of $x$'s coordinates. Fortunately, these two complexities can only differ by $O(\log r)$, and we will frequently switch between the two types of approximations in our proofs.

For $r\in\N$, let $\D_r=2^{-r}\Z$ denote the \emph{$r$-dyadic rationals}, and for $x=(x_1,\ldots,x_n)\in\R^n$, let
\[\lfloor x\rfloor_r=(2^{-r}\lfloor 2^r x_1\rfloor,\ldots,2^{-r}\lfloor 2^r x_n\rfloor)\in\D_r^n.\]
For $e\in\Sphere^1$, we identify $e$ with the real number $a\in[0,1)$ such that $e=(\cos(2\pi a),\sin(2\pi a))$, so $\lfloor e\rfloor_r$ denotes $\lfloor a\rfloor_r$ in this case.

Let $C\in\N$ be a constant such that, for all $r\in\N$, $\sigma\in\{0,1\}^*$, and $x\in\R\cup\R^2\cup\Sphere^1$,
\begin{equation}\label{eq:dyadic}
    |K_r(x\mid \sigma)-K(\lfloor x\rfloor_r\mid \sigma)|<C\log r
\end{equation}
and
\begin{equation}\label{eq:dyadiccond}
    |K(\sigma\mid \mathcal{B}_{2^{-r}}(x))-K(\sigma\mid \lfloor x\rfloor_r)|<C\log r.
\end{equation}
For details on the existence of this constant, see Section 2.3 of~\cite{LutStu20}.

A closely related useful fact is that the complexity of an approximation is only linearly sensitive to changes in the level of precision. In particular, it was shown by~\cite{CasLut15,LutLut18} that for each $n\in\N$, there is a constant $C_0\in\N$ such that for all $x,y\in\R^n$ and all $r,s,t\in\N$,
\begin{equation}\label{eq:precision}
    K_{r+t,s}(x\mid y)\leq K_{r,s}(x\mid y)+nt+C_0\log(r+s+t)
\end{equation}
and
\begin{equation}\label{eq:condprecision}
    K_{r,s}(x\mid y)\leq K_{r,s+t}(x\mid y)+nt+C_0\log(r+s+t).
\end{equation}

\subsection{Hausdorff Dimension and Effective Hausdorff Dimension}\label{ssec:psp}
For $E\subseteq\R^n$ and $s>0$, the \emph{$s$-dimensional Hausdorff measure} of $E$ is
\[\mathcal{H}^s(E)=\lim_{\delta\to 0^+}\inf\left\{\sum_{i\in\N}|U_i|^s \;\middle|\; E\subseteq\bigcup_{i\in\N}U_i\text{ and each }|U_i|<\delta\right\},\]
where $|U_i|$ denotes the diameter of $U_i$, and the \emph{Hausdorff dimension} of $E$ is
\[\dimH E=\inf\{s\mid \mathcal{H}^s(E)=0\}.\]
The statement of Theorem~\ref{thm:pinnedclassical}, like many classical results on Hausdorff dimension, requires the underlying set to be \emph{analytic}, i.e., a continuous image of a Borel set.

The \emph{(effective Hausdorff) dimension} of an individual point $x\in\R^n$ is defined~\cite{DISS,KCCCHD} as
\[\dim(x)=\liminf_{r\to\infty}\frac{K_r(x)}{r}.\]

As $U$ was an oracle prefix machine, we can readily define Kolmogorov complexities \emph{relative} to any oracle $A\subseteq\N$. For example, for all $\sigma,\tau\in\{0,1\}^*$,
\[K^A(\sigma\mid\tau)=\min\left\{|\pi|\mid \pi\in\{0,1\}^*\text{ and }U^A(\pi,\tau)=\sigma\right\}.\]
All other Kolmogorov complexity quantities in this section, including $\dim(x)$, are relativized in the same way, always denoted by placing the oracle in the superscript. The bounds~\eqref{eq:soi}--\eqref{eq:condprecision} all hold relative to an arbitrary oracle $A\subseteq\N$. Given a point $y\in\R^n$, we will write $K^y$ and $\dim^y$ as shorthand for $K^{A_y}$ and $\dim^{A_y}$, where $A_y\subseteq\N$ is an oracle that encodes $\lfloor y\rfloor_r$ for all $r\in\N$. When multiple oracle sets are placed in the superscript, that indicates relativization to their Turing join.

Our Hausdorff dimension bounds will rely on the \emph{point-to-set principle}~\cite{LutLut18}, which states that, for all $E\subseteq\R^n$,
\begin{equation}\label{eq:psp}
    \dimH E=\adjustlimits\min_{A\subseteq\N}\sup_{x\in E}\dim^A(x).
\end{equation}
An oracle $A$ that achieves the minimum in~\eqref{eq:psp} is called a \emph{Hausdorff oracle} for $E$. We will also use the correspondence principle of Hitchcock~\cite{Hitchcock05}, which says that, if $E\subseteq\R^n$ is computably compact relative to an oracle $A$, then $A$ is a Hausdorff oracle for $E$. That is,
\begin{equation}\label{eq:pspCompact}
    \dimH E = \sup_{x\in E}\dim^A(x).
\end{equation}

A Hausdorff oracle $A$ for $E$ is an \emph{optimal Hausdorff oracle} for $E$ if, for every auxiliary oracle $B\subseteq\N$ and every $\varepsilon>0$, there is some $x\in E$ such that
\[\dim^{A,B}(x)\geq\dimH(E)-\varepsilon\]
and, for almost every $r\in\N$,
\[K_r^{A,B}(x)\geq K_r^A(x)-\varepsilon r.\]
Stull~\cite{Stull22a} formulated this definition and proved that if $E$ is an analytic set, then there are optimal Hausdorff oracles for $E$, and that the converse does not hold.

\section{Combinatorial Lemma}\label{sec:comb}

The proofs of each of our main Kolmogorov complexity bounds will rely on finding objects $u$, $v$, and $d$ that satisfy two criteria: First, $d$ is in some sense ``suitable'' for each of $u$ and $v$, belonging to certain sets $N_u$ and $N_v$. Second, $u$ and $v$ are sufficiently dissimilar with respect to $d$, satisfying $u\not\sim_d v$, where $\sim_d$ is a binary relation that may depend on $d$. The following lemma says that if every $v$ has many suitable $d$ and, for each such $d$, there are few $u$ such that $u\sim_d v$, then there is some pair $u,v$ such that many $d$ satisfy these two criteria for that pair.

\begin{lemma}\label{lem:comb}
    Let $X$ and $V$ be finite sets, $\alpha\in(0,1/2)$, and, for each $d\in X$, $\sim_d$ a binary relation on $V$. Suppose there is a collection $\{N_v\mid v\in V\}$ of finite subsets of $X$ such that, for all $v\in V$, we have $|N_v|\geq\alpha|X|$ and, for all $d\in N_v$,
    \begin{equation}\label{eq:fewsim}
        |\{u\in V\mid u\sim_d v\}|\leq \frac{\alpha^2}{4}|V|.
    \end{equation}
    Then there exist $u,v\in V$ such that $|\{d\in X\mid d\in N_u\cap N_v\text{ and }u\not\sim_d v\}|>\frac{\alpha^2}{2}|X|$.
    In particular, if there is a single binary relation $\sim$ such that $\sim_d$ is $\sim$ for all $d\in X$, then there exist $u,v\in V$ such that $u\not\sim v$ and $|N_u\cap N_v|>\frac{\alpha^2}{2}|X|$.
\end{lemma}
\begin{proof}
    Assume toward a contradiction that $|\{d\in X\mid d\in N_u\cap N_v\text{ and }u\not\sim_d v\}|\leq\frac{\alpha^2}{2}|X|$
    holds for all $u,v\in V$. Then we can trivially bound the cardinality of the set
    \[B=\{(d,u,v)\mid d\in N_u\cap N_v,\ u,v\in V,\ \text{and}\ u\not\sim_d v.\}\]
    by $|B|\leq \frac{\alpha^2}{2}|V|^2|X|$.

    We now derive a lower bound on $|B|$ to contradict this upper bound. For every $d\in X$, let $n_d=|\{v\in V\mid d\in N_v\}|$,
    and define the set $\mathcal{E}=\{d\in X\mid n_d>\alpha^2|V|/4\}$. Trivially,
    \begin{equation}\label{eq:sumndupper}
        \sum_{d\in\mathcal{E}}n_d\leq |V||X|
    \end{equation}
    and we can bound this sum from below by
    \begin{align}
        \sum_{d\in \mathcal{E}}n_d&=\sum_{d\in X}n_d-\sum_{d\in X\setminus\mathcal{E}}n_d\notag\\
        &=\sum_{v\in V}|N_v|-\sum_{d\in X\setminus\mathcal{E}}n_d\notag\\
        &\geq\alpha|V||X|-\sum_{d\in X\setminus\mathcal{E}}n_d\notag\\
        &\geq (\alpha-\alpha^2/4)|V||X|.\label{eq:sumndlower}
    \end{align}
    Therefore,
    \begin{align*}
        |B|&=\sum_{d\in X}\left|\left\{(u,v)\in V^2\;\middle\vert\; d\in N_u\cap N_v\text{ and }u\not\sim_d v\right\}\right|\\
        &\geq \sum_{d\in\mathcal{E}} n_d\cdot\left(n_d-\frac{\alpha^2}{4}|V|\right)\tag{by~\eqref{eq:fewsim}}\\
        &=\sum_{d\in\mathcal{E}} n_d^2-\frac{\alpha^2}{4}|V|\sum_{d\in\mathcal{E}} n_d\\
        &\geq\sum_{d\in\mathcal{E}} n_d^2-\frac{\alpha^2}{4}|V|^2|X|\tag{by~\eqref{eq:sumndupper}}\\
        &\geq \frac{1}{|X|}\bigg(\sum_{d\in\mathcal{E}}n_d\bigg)^2-\frac{\alpha^2}{4}|V|^2|X|\tag{by Cauchy--Schwarz}\\
        &\geq\frac{((\alpha-\alpha^2/4)|V||X|)^2}{|X|}-\frac{\alpha^2}{4}|V|^2|X|\tag{by~\eqref{eq:sumndlower}}\\
        &>\frac{\alpha^2}{2}|V|^2|X|,\tag{$\alpha\in(0,1/2)$}
    \end{align*}
    contradicting our upper bound on $|B|$.
\end{proof}

\section{Kolmogorov Complexity Bound for Distances}

The proof of our first main Kolmogorov complexity bound, Theorem~\ref{thm:pinned}, relies on triangulating the approximate location of a point $d$ using its approximate distance from two other points $u$ and $v$. Geometrically, this places $d$ in the intersection of two thin annuli, centered at $u$ and $v$. We will use the following lemma of Wolff~\cite{Wolff97} to bound the size of that intersection.

\begin{lemma}[Wolff~\cite{Wolff97}]\label{lem:wolff}
    Let $\rho_1,\rho_2\in[0.99,1.01]$, $x_1,x_2\in [0,0.01]^2$, and $\varepsilon>0$. Let $A_1$ and $A_2$ be the open $\varepsilon$-neighborhoods of the circles with centers $x_1$ and $x_2$ and radii $\rho_1$ and $\rho_2$, respectively. Then $A_1\cap A_2$ is contained within $O(1)$ sets of the form
    \[\{y\in A_1\mid x_1+\rho_1 e_{x_1,y}\in\gamma\},\]
    where each $\gamma$ is an arc of length
    \[\frac{2\pi\rho_1\varepsilon}{\sqrt{\max(|x_1-x_2|+|\rho_1-\rho_2|,\varepsilon)\max(||x_1-x_2|-|\rho_1-\rho_2||,\varepsilon)}}.\]
\end{lemma}

For all distinct points $u,v\in\R^2$, let $e_{u,v}$ denote $\frac{v-u}{|v-u|}\in\Sphere^1$, the direction from $u$ to $v$.

\begin{theorem}\label{thm:pinned}
    If $A\subseteq\N$, $x,y\in\R^2$, $r\in\N$, and constants $\delta,\varepsilon\in\Q\cap(0,1)$ satisfy $\delta<\frac{\varepsilon}{6}\cdot\dim^{A,y}(e_{y,x})$,
    \begin{equation}\label{eq:x_pinned}
        K_r^{A,y}(x)\geq K_r^A(x)-\delta r,
    \end{equation}
    and $r$ is sufficiently large, then for all $s\in\N$,
    \begin{equation}\label{eq:pinned}
        K_{r,s}^{A,y}(|x-y|\mid |x-y|)\geq \frac{K_{r,s}^A(x\mid x)}{2}-\varepsilon r.
    \end{equation}
\end{theorem}

Before giving the proof, we note that the primary strength of this theorem is the weakness of the assumptions. Namely, the theorem does not require that $x,y$ or $e_{y,x}$ are highly complex. The only assumptions are that the dimension of $e_{y,x}$ is non-zero and that $x$ and $y$ are sufficiently independent. This makes Theorem \ref{thm:pinned} incomparable to the previous results on the algorithmic complexity of distances, such as Theorem 22 of \cite{Stull22c}, which all require strong assumptions on the randomness of $x,y$ and $e_{y,x}$.

\begin{proof}
    Let $A$, $x$, $y$, $r$, $s$, $\delta$, and $\varepsilon$ satisfy the hypotheses. We suppress the oracle $A$ from the notation below, but the argument is unchanged by working relative to $A$. There is a computable affine transformation $f:\R^2\to\R^2$ with computable Lipschitz constants such that $\lfloor f(x)\rfloor_7=(1,0)$ and $\lfloor f(y)\rfloor_7=(0,0)$. As this transformation is computable and bi-Lipschitz, applying $f$ can only change complexities at precision $r$ by $O(\log r)$~\cite{CasLut15,LutLut18}, so we assume without loss of generality that $\lfloor x\rfloor_7=(1,0)$ and $\lfloor y\rfloor_7=(0,0)$; this assumption will make it more convenient to apply Lemma~\ref{lem:wolff} as written, since $2^{-7}<0.01$.

    When $s\geq r-\varepsilon r$, the right-hand side of~\eqref{eq:pinned} is negative for all sufficiently large $r$. Hence, assume $s<r-\varepsilon r$, and assume toward a contradiction that~\eqref{eq:pinned} does not hold. We now construct sets $X$, $V$, and $\{N_v\mid v\in V\}$, to which we will apply Lemma~\ref{lem:comb}.
    
    \paragraph*{Construction}        
        For all $v,d\in\R^2$, denote by $v\rhd d$ the inequality
        \[K_{r,s}^{v}(|d-v|\mid |d-v|)<\frac{K_{r,s}(x\mid x)}{2}-\varepsilon r+4C_0\log r,\]
        where $C_0$ is the constant from~\eqref{eq:precision}. Intuitively, $v\rhd d$ means that $v$ is a suitable point to use in triangulating the location of the point $d$. Our goal will be to find two such points $u$ and $v$ such that $u$, $v$, and $d$ aren't too close to being collinear.

        Define the sets
        \begin{align*}
            X&=\left\{d\in\D_r^2\;\middle|\;
            \begin{array}{l}
                \lfloor d \rfloor_7=(1,0),\ K(d)\leq K_r(x)+C\log r,\\
                K(d\mid\lfloor d\rfloor_s)\leq K_{r,s}(x\mid x)+2C\log r
            \end{array}\right\}\\
            Y&=\left\{d\in\D_r^2\;\middle|\;
                \lfloor d\rfloor_7=(0,0),\ K(d)\leq K_r(y)+C\log r
            \right\},
        \end{align*}
        where $C$ is the constant from~\eqref{eq:dyadic}. Let $t=\lfloor\varepsilon r/3\rfloor$, noting that this is computable given $r$. Let $z=(r,s,K_r(x),K_{r,s}(x\mid x),K_r(y),K^y_t(e_{y,x}))$, noting that $K(z)=O(\log r)$ by~\eqref{eq:nat}, and let $C_3,C_4,C_5\in\N$ be constants to be defined later.

        Let $M$ be an oracle prefix Turing machine that, given $v$ as an oracle and a program for $z$ as an input, internally enumerates all $d\in X$ such that $v\rhd d$. For each $e\in\D_t$, whenever $M$ has internally enumerated exactly
        \begin{equation}\label{eq:crowded_pinned}
            \big\lfloor 2^{K_r(x)-K_t^y(e_{y,x})-\delta r-C_3\log r}\big\rfloor
        \end{equation}
        points $d$ such that $\lfloor e_{v,d}\rfloor_t=e$, the machine labels $e$ ``crowded for $v$'' and prints those points.

        For each $v\in Y$, let $F_v$ be the set of all directions that are crowded for $v$. Define the set
       \begin{equation}\label{eq:V1_pinned}
            V_1=\left\{d\in Y\;\middle\vert\; (\exists v\in\R^2)\left[\lfloor v\rfloor_r=d\text{ and }|F_v|>2^{K^y_t(e_{y,x})-C_4\log r}\right]\right\}.
        \end{equation}
        For each $d\in V_1$, let $v_d\in\R^2$ be a point satisfying $\lfloor v_d\rfloor_r=d$ and $|F_{v_d}|>2^{K_t^y(e_{y,x})-C_4\log r}$, and define the set $V=\{v_d\mid d\in V_1\}$.

        For each $v\in V$, define
        \[N_v=\left\{d\in X\mid M^v(\pi)\text{ prints }d\text{ and }K^v(\lfloor e_{v,d}\rfloor_t)\geq K^y_t(e_{y,x})-C_4\log r-2\right\},\]
        where $\pi$ is a witness to $K(z)$ and $C_4\in\N$ is the constant from~\eqref{eq:V1_pinned}.
        
    We will prove that the hypotheses of Lemma~\ref{lem:comb} are satisfied by the sets $X$ and $V$, the collection $\{N_v\mid v\in V\}$, the parameter $\alpha=2^{-\delta r-(C_3+C_4)\log r-3}$, and the collection $\{\sim_d\;\mid d\in X\}$ of reflexive binary relations defined for $u\neq v$ by $u\sim_d v$ if $|e_{u,d}-e_{v,d}|<2^{-t}$. We first state and prove three preliminary claims.
    \begin{claim}
        \begin{equation}\label{eq:X_pinned}
            2^{K_r(x)-O(\log r)}\leq |X|< 2^{K_r(x)+C\log r+1}.
        \end{equation}
    \end{claim}
    \begin{proof}
        Recall that $s\leq r$, and note that $X$ is uniformly computably enumerable given $z$. Furthermore, $\lfloor x\rfloor_r\in X$; it is immediate from~\eqref{eq:dyadic} that $K(\lfloor x\rfloor_r)\leq K_r(x)$, and
        \begin{align*}
            K(\lfloor x\rfloor_r\mid\lfloor\lfloor x\rfloor_r\rfloor_s)&= K(\lfloor x\rfloor_r\mid\lfloor x\rfloor_s)\\
            &<K(\lfloor x\rfloor_r\mid\mathcal{B}_{2^{-s}}(x))+C\log r\tag{by~\eqref{eq:dyadiccond}}\\
            &< K_{r,s}(x\mid x)+2C\log r.\tag{by~\eqref{eq:dyadic}}
        \end{align*}
        As $\lfloor x\rfloor_r\in X$, we have
        \begin{align*}
            K_r(x)&\leq K(\lfloor x\rfloor_r)+C\log r\tag{by~\eqref{eq:dyadic}}\\
            &\leq K(\lfloor x\rfloor_r\mid z)+O(\log r)\tag{by~\eqref{eq:soi}}\\
            &\leq \log|X|+O(\log r),\tag{by~\eqref{eq:cardcomp}}
        \end{align*}
        which give the lower bound on $|X|$ in~\eqref{eq:X_pinned}. For the upper bound, note that there are fewer than $2^{K_r(x)+C\log r+1}$ objects with complexity at most $K_r(x)+C\log r$.
    \end{proof}

    \begin{claim}
        \begin{equation}\label{eq:Fy}
            |F_y|\geq 2^{K^y_t(e_{y,x})-C_4\log r}.
        \end{equation}
    \end{claim}
    \begin{proof}
        Since $|\lfloor x\rfloor_r-x|<2^{1/2-r}\leq 2^{1/2-s}$,
        we have
        \begin{align*}
            K_{r,s}^{y}(|\lfloor x\rfloor_r-y|\mid |\lfloor x\rfloor_r-y|)&\leq K_{r-2,s+2}^{y}(|\lfloor x\rfloor_r-y|\mid |\lfloor x\rfloor_r-y|)+4C_0\log r\tag{by~\eqref{eq:precision} and~\eqref{eq:condprecision}}\\
            &\leq K_{r,s}^{y}(|x-y|\mid |x-y|)+4C_0\log r\\
            &<\frac{K_{r,s}(x\mid x)}{2}-\varepsilon r+4C_0\log r,\tag{by assumption}
        \end{align*}
        so $y\rhd\lfloor x\rfloor_r$. Let
        \[T=\{d\in X\mid y\rhd d\text{ and }\lfloor e_{y,d}\rfloor_t=\lfloor e_{y,\lfloor x\rfloor_r}\rfloor_t\}.\]
        Clearly $\lfloor x\rfloor_r\in T$, and $T$ is uniformly computably enumerable relative to $y$, given $\lfloor e_{y,\lfloor x\rfloor_r}\rfloor_t$ and $z$. Therefore,
        \begin{align*}
            K_r(x)-K_t^y(e_{y,x})&\leq K^y_r(x)-K_t^y(e_{y,x})+\delta r+O(\log r)\tag{by~\eqref{eq:x_pinned}}\\
            &\leq K^y_r(x\mid \lfloor e_{y,\lfloor x\rfloor_r}\rfloor_t,z)+\delta r+O(\log r)\tag{by~\eqref{eq:precision},~\eqref{eq:soi},and~\eqref{eq:dyadic}}\\
            &\leq \log|T|+\delta r+O(\log r)\tag{by~\eqref{eq:cardcomp}},
        \end{align*}
        so there is some constant $C_3\in\N$ such that
        \[|T|\geq 2^{K_r(x)-K_t^y(e_{y,x})-\delta r-C_3\log r}.\]
        Choosing this $C_3$ as the constant in~\eqref{eq:crowded_pinned}, we have shown that
        \[\lfloor e_{y,\lfloor x\rfloor_r}\rfloor_t\in F_y.\]
        As $F_y$ is uniformly computably enumerable relative to $y$, given $z$, we have
        \begin{align*}
            K^y_t(e_{y,x})&\leq K^y(\lfloor e_{y,\lfloor x\rfloor_r}\rfloor_t)+O(\log r)\tag{by~\eqref{eq:precision} and~\eqref{eq:dyadic}}\\
            &\leq \log|F_y|+O(\log t),
        \end{align*}
        so there is some constant $C_4\in\N$ such that~\eqref{eq:Fy} holds.
    \end{proof}

    \begin{claim}
        \begin{equation}\label{eq:vcard}
            |V|\geq 2^{K_r(y)-O(\log r)}.
        \end{equation}
    \end{claim}
    \begin{proof}
        First note that $Y$ is uniformly computably enumerable given $z$ and that~\eqref{eq:dyadic} implies $\lfloor y\rfloor_r\in Y$. It follows that $V_1$ is also uniformly computably enumerable given $z$, and by~\eqref{eq:Fy}, $\lfloor y\rfloor_r\in V_1$. Therefore,
        \begin{align*}
            K_r(y)&\leq K(\lfloor y\rfloor_r)+C\log r\tag{by~\eqref{eq:dyadic}}\\
            &\leq K(\lfloor y\rfloor_r\mid z)+K(z)+C\log r\tag{by~\eqref{eq:soi}}\\
            &=\log|V_1|+O(\log r).\tag{by~\eqref{eq:nat} and~\eqref{eq:cardcomp}}
        \end{align*}
        so $|V_1|\geq 2^{K_r(y)-O(\log r)}$. As $|V|=|V_1|$, this proves~\eqref{eq:vcard}.
    \end{proof}
    \begin{claim}
        The hypotheses of Lemma~\ref{lem:comb} are satisfied by the sets $X$ and $V$, the collection $\{N_v\mid v\in V\}$, the parameter $\alpha=2^{-\delta r-(C_3+C_4)\log r-3}$,
        and the collection $\{\sim_d\;\mid d\in X\}$ of reflexive binary relations defined for $u\neq v$ by $u\sim_d v$ if $|e_{u,d}-e_{v,d}|<2^{-t}$.
    \end{claim}
    \begin{proof}
        Fix any $v\in V$ and $d\in N_v$. By~\eqref{eq:V1_pinned}, there are at least $2^{K_t^y(e_{y,x})-C_4\log r}$ directions $e\in\D_t$ that are crowded for $v$. As there are fewer than $2^{K_t^y(e_{y,x})-C_4\log r-1}$ objects with complexity less than $K_t^y(e_{y,x})-C_4\log r-2$ relative to $v$, there are more than
        $2^{K_t^y(e_{y,x})-C_4\log r-1}$ directions $e\in\D_t$ that are crowded for $v$ and satisfy $K^v(e)\geq K_t^y(e_{y,x})-C_4\log r-2$. For each of these directions, $M^v(\pi)$ prints exactly $\big\lfloor 2^{K_r(x)-K_t^y(e_{y,x})-\delta r-C_3\log r}\big\rfloor$ points $d$, so
        \begin{align*}
            |N_v|&>2^{K_t^y(e_{y,x})-C_4\log r-1} \big\lfloor 2^{K_r(x)-K_t^y(e_{y,x})-\delta r-C_3\log r}\big\rfloor\\
            &\geq 2^{K_r(x)-\delta r-(C_3+C_4)\log r-2}\\
            &> \alpha|X|.\tag{by~\eqref{eq:X_pinned}}
        \end{align*}
        
        For~\eqref{eq:fewsim}, observe that for every $C_5\in\N$, there can be only $2^{K_t^y(e_{y,x})-C_5\log r+O(1)}$ directions $e\in \D_t$ such that
        \[|\{u\in V \mid |e_{v,u}-e|<2^{-t}\}|\geq 2^{-K_t^y(e_{y,x})+C_5\log r}|V|.\]
        The set of all such directions $e$ is uniformly computably enumerable relative to $v$, given $z$, so if $\lfloor e_{v,d}\rfloor_t$ were such a direction, then we would have
        \begin{align*}
            K^v(\lfloor e_{v,d}\rfloor_t)&\leq K^v(\lfloor e_{v,d}\rfloor_t\mid z)+O(\log r)\tag{by~\eqref{eq:soi}}\\
            &\leq K_t^y(e_{y,x})-C_5\log r+O(1)+O(\log r)\tag{by~\eqref{eq:cardcomp}}\\
            &<K_t^y(e_{y,x})-C_4\log r-2,
        \end{align*}
        for every sufficiently large constant $C_5\in\N$. Fixing such a $C_5$, this contradicts the membership of $d$ in $N_v$. Furthermore, as $V\subseteq[0,2^{-7})^2$ and $\lfloor d\rfloor_7=(1,0)$, if $|e_{v,u}-e_{v,d}|<2^{-t}$ then we also have $|e_{u,d}-e_{v,d}|<2^{-t}$. Therefore,
        \begin{align*}
            |\{u\in V: u\sim_d v\}|&\leq |\{u\in V: |e_{v,u}-e_{v,d}|<2^{-t}\}|\\
            &<2^{-K_t^y(e_{y,x})+C_5\log r}|V|\\
            &\leq 2^{-t\dim^y(e_{y,x})+o(t)+C_5\log r}|V|\\
            &=2^{-\varepsilon r\dim^y(e_{y,x})/3+o(r)}\tag{$t=\lfloor\varepsilon r/3\rfloor$}\\
            &<2^{-2(\delta r+(C_3+C_4)\log r+3)-2}|V|\\
            &=\frac{\alpha^2}{4}|V|
        \end{align*}
        holds when $r$ is sufficiently large, as $\delta<\varepsilon\dim^y(e_{y,x})/6$. As $v\in V$ and $d\in N_v$ were arbitrary, we conclude that~\eqref{eq:fewsim} holds.
    \end{proof}

    Hence, we can apply Lemma~\ref{lem:comb} to let $u,v\in V$ be such that
    \[|\{d\in X\mid d\in N_u\cap N_v\text{ and }|e_{u,d}-e_{v,d}|\geq 2^{-t}\}|>2^{K_r(x)-2\delta r-O(\log r)}.\]
    Since this set is large, it contains high-complexity points relative to every oracle; in particular, let $d\in N_u\cap N_v$ such that $|e_{u,d}-e_{v,d}|\geq 2^{-t}$ and
    \begin{equation}\label{eq:Kuvd_pinned}
        K^{u,v}(d)\geq K_r(x)-2\delta r-O(\log r).
    \end{equation}
    
    Let $\pi_u$ be a witness to $K^u(\lfloor |d-u|\rfloor_r \mid \lfloor |d-u|\rfloor_s)$ and $\pi_v$ be a witness to $K^v(\lfloor |d-v|\rfloor_r\mid \lfloor |d-v|\rfloor_s)$, and define the sets $P_u=\{U^u(\pi_u,d_u)\mid d_u\in\D_s\}$ and $P_v=\{U^v(\pi_v,d_v)\mid d_v\in\D_s\}$. Intuitively, each of these sets contains, within each $s$-dyadic interval, at most one candidate value for an approximation of a distance from $d$. In particular, $\lfloor |d-u|\rfloor_r\in P_u$ and $\lfloor |d-v|\rfloor_r\in P_v$. Let $D_{u,v}$ be the set of all $r$-dyadic points that share these properties; that is,
    \[D_{u,v}=\left\{w\in\D_r^2\mid \lfloor|w-u|\rfloor_r\in P_u\text{ and }\lfloor |w-v|\rfloor_r\in P_v\right\}.\]
    For any two points $w,w'\in D_{u,v}$ such that $|w-w'|<2^{1-s}$, we have
    \[|\lfloor|w-u|\rfloor_r -\lfloor |w'-u|\rfloor_r|\leq \big||w-u|-|w'-u|\big|+2^{1-r}< 2^{2-s},\]
    and similarly $|\lfloor|w-v|\rfloor_r -\lfloor |w'-v|\rfloor_r|<2^{2-s}$. As each $s$-dyadic interval contains at most one element of $P_u$ and one element of $P_v$, this implies that the set
    \[P_{u,v}=\{(\lfloor |w-u|\rfloor_r,\lfloor |w-v|\rfloor_r)\mid w\in D_{u,v}\text{ and }\lfloor w\rfloor_s=\lfloor d\rfloor_s\}\]
    has cardinality $O(1)$.
    
    Now consider the set
    \[\hat{D}_{u,v}=\{w\in\D_r^2\mid \lfloor w-u\rfloor_r=\lfloor d-u\rfloor_r\text{ and }\lfloor w-v\rfloor_r=\lfloor d-v\rfloor_r\}.\]
    This set is contained in the intersection of $O(2^{-r})$-neighborhoods of the circles with centers $u$ and $v$ and radii $\rho_u:=|d-u|$ and $\rho_v:=|d-v|$, respectively. As $|e_{u,d}-e_{v,d}|\geq 2^{-t}$, for all sufficiently large $r$ we have
    \[\max\{|u-v|+|\rho_u-\rho_v|,O(2^{-r})\}\geq 2^{-t-O(1)}\]
    and, by a straightforward trigonometric argument (also using $|u-v|<2^{-8}$ and $\rho_u,\rho_v> 1-2^{-7}$),
    \[\max\{||u-v|-|\rho_u-\rho_v||,O(2^{-r})\}\geq 2^{-3t-O(1)}.\]
    Therefore, by Lemma~\ref{lem:wolff}, $\hat{D}_{u,v}$ is covered by $O(1)$ sets, each of diameter
    \[\frac{O(2^{-r})}{\sqrt{2^{-4t}}}=O(2^{2t-r}).\]
    
    It follows that some point in $\hat{D}_{u,v}$ is a precision-$(r-2t-O(1))$ approximation of $d$ that can be computed, relative to $u$ and $v$, given $\pi_u$, $\pi_v$, a witness to $K(\lfloor d\rfloor_s)$, $O(1)$ bits to specify $(|d-u|\rfloor_r,\lfloor |d-v|\rfloor_r)$ within the set $P_{u,v}$, $O(1)$ bits to specify a specific cover element, and $O(1)$ bits for the description of the Turing machine that performs this computation. That is,
    \begin{equation}\label{eq:compute_d_pinned}
        K^{u,v}_{r-2t-O(1)}(d)\leq |\pi_u|+|\pi_v|+K(\lfloor d\rfloor_s)+O(1).
    \end{equation}
    Therefore,
    \begin{align*}
        K_{r,s}(x\mid x)&=K(\lfloor x\rfloor_r\mid \lfloor x\rfloor_s)+O(\log r)\tag{by~\eqref{eq:dyadic} and~\eqref{eq:dyadiccond}}\\
        &= K(\lfloor x\rfloor_r)-K(\lfloor x\rfloor_s)+O(\log r)\tag{by~\eqref{eq:soi}}\\
        &=K_r(x)-K(\lfloor x\rfloor_s)+O(\log r)\tag{by~\eqref{eq:dyadic}}\\
        &\leq K^{u,v}(d)-K_s(x)+2\delta r+O(\log r)\tag{by~\eqref{eq:Kuvd_pinned}}\\
        &= K_r^{u,v}(d)-K(\lfloor d\rfloor_s)+2\delta r+O(\log r)\tag{by~\eqref{eq:dyadic}}\\
        &\leq K^{u,v}_{r-2t-O(1)}(d)+4t-K(\lfloor d\rfloor_s)+2\delta r+O(\log r)\tag{by~\eqref{eq:precision}}\\
        &\leq |\pi_u|+|\pi_v|+4t+2\delta r+O(\log r)\tag{by~\eqref{eq:compute_d_pinned}}\\
        &= K_{r,s}^u(|d-u|\mid |d-u|)+K_{r,s}^v(|d-v|\mid |d-v|)+4t+2\delta r+O(\log r)\tag{by~\eqref{eq:dyadic} and~\eqref{eq:dyadiccond}}\\
        &\leq K_{r,s}(x\mid x)-2\varepsilon r+4t+2\delta r+O(\log r)\tag{$d\in N_u\cap N_v$}\\
        &<K_{r,s}(x\mid x)-2\varepsilon r/3+2\delta r+O(\log r).\tag{$t\leq \varepsilon r/3$}
    \end{align*}
    As $\delta<\varepsilon\dim^y(e_{y,x})/6\leq\varepsilon/6$, this implies $\varepsilon r/3< O(\log r)$, a contradiction for sufficiently large $r$. We conclude that~\eqref{eq:pinned} holds.
\end{proof}

\section{Kolmogorov Complexity Bound for Orthogonal Projections}

\begin{theorem}\label{thm:orth}
    If $A\subseteq\N$, $x\in\R^2$, $e\in\Sphere^1$, $r\in\N$, and constants $\delta,\varepsilon\in\Q\cap(0,1)$ satisfy $\delta<\varepsilon\dim^A(e)/4$,
    \begin{equation}\label{eq:x}
        K_r^{A,e}(x)\geq K^A_r(x)-\delta r,
    \end{equation}
    and $r$ is sufficiently large, then for all $s\in\N$,
    \begin{equation}\label{eq:orth}
        K_{r,s}^{A,e}(p_e x\mid p_e x)\geq \frac{K^A_{r,s}(x\mid x)}{2}-\varepsilon r.
    \end{equation}
\end{theorem}
\begin{proof}
    Let $A$, $x$, $e$, $r$, $s$, $\delta$, and $\varepsilon$ satisfy the hypotheses. We suppress the oracle $A$ from the notation below, but the argument is unchanged by working relative to $A$. When $s\geq r-\varepsilon r$, the right-hand side of~\eqref{eq:orth} is negative for all sufficiently large $r$. Hence, assume $s<r-\varepsilon r$, and assume toward a contradiction that~\eqref{eq:orth} does not hold. We now construct sets $X$ and $V$ and a collection $\{N_v\mid v\in V\}$, to which we will apply Lemma~\ref{lem:comb}.
    
    \paragraph*{Construction}
    Let $t=\lfloor\varepsilon r/2\rfloor$, noting that this is computable given $r$, and define the sets
    \begin{align*}
        X&=\{d\in\D_r^2\mid K(d)\leq K_r(x)+C\log r\ \text{and}\ K(d\mid\lfloor d\rfloor_s)\leq K_{r,s}(x\mid x)+2C\log r\}\\
        S&=\{d\in\D_r\mid K(d)\leq K_r(e)+C\log r\ \text{and}\ K_t(d)\leq K_t(e)+2C\log t\},
    \end{align*}
    where $C$ is the constant from~\eqref{eq:dyadic}.
    
    For every direction $v\in\Sphere^1$, define the set
    \[N_v=\left\{d\in X\;\middle\vert\;K_{r,s}^v(p_vd\mid p_vd)<\frac{K_{r,s}(x\mid x)}{2}-\varepsilon r+4C_0\log r\right\},\]
    where $C_0$ is the constant from~\eqref{eq:precision}.

    Let $z=(r,s,K_r(x),K_{r,s}(x\mid x),K_r(e),K_t(e))$, noting that $K(z)=O(\log r)$ by~\eqref{eq:nat}, and let $C_1,C_2\in\N$ be constants to be defined later. Define the set
    \begin{equation}\label{eq:V1}
        V_1=\left\{d\in S\;\middle\vert\; (\exists v\in\Sphere^1)\left[\lfloor v\rfloor_r=d\text{ and }|N_v|>2^{K_r(x)-\delta r-C_1\log r}\right]\right\}.
    \end{equation}

    Let $M$ be a prefix Turing machine that, given any input $\pi\in\{0,1\}^*$ such that $U(\pi)\in\N^6$, does the following. Working under the assumption that $U(\pi)=z$, $M$ internally enumerates $V_1$. For each $d\in\D_t$, whenever $M$ has internally enumerated exactly $\lfloor 2^{K_r(e)-K_t(e)-C_2\log r}\rfloor$ directions $d'\in V_1$ such that $\lfloor d'\rfloor_t=d$, the machine labels $d$ ``crowded'' and prints these $\lfloor 2^{K_r(e)-K_t(e)-C_2\log r}\rfloor$ directions $d'$.

    Let $V_2$ be the set of all directions printed by $M(\pi)$, where $\pi$ is a witness to $K(z)$. As $V_2\subseteq V_1$, for each $d\in V_2$ there is some $v_d\in\Sphere^1$ satisfying $\lfloor v_d\rfloor_r=d$ and $|N_{v_d}|>2^{K_r(x)-\delta r-C_1\log r}$; let $V=\{v_d\mid d\in V_2\}$.

    We will show that the sets $X$ and $V$, the collection $\{N_v\mid v\in V\}$, $\alpha=2^{-\delta r-(C+C_2)\log r-1}$, and the binary relation $|u-v|<2^{-t}$ satisfy the hypotheses of Lemma~\ref{lem:comb}. As in the proof of Theorem~\ref{thm:pinned}, we first state and prove three preliminary claims.

    \begin{claim}
        \begin{equation}\label{eq:X}
            2^{K_r(x)-O(\log r)}\leq |X|< 2^{K_r(x)+C\log r+1}.
        \end{equation}
    \end{claim}
    \begin{proof}
        First note that $X$ is uniformly computably enumerable given $z$. Furthermore, $\lfloor x\rfloor_r\in X$; it is immediate from~\eqref{eq:dyadic} that $K(\lfloor x\rfloor_r)\leq K_r(x)$, and we have
        \begin{align*}
            K(\lfloor x\rfloor_r\mid\lfloor\lfloor x\rfloor_r\rfloor_s)&= K(\lfloor x\rfloor_r\mid\lfloor x\rfloor_s)\\
            &<K(\lfloor x\rfloor_r\mid\mathcal{B}_{2^{-s}}(x))+C\log r\tag{by~\eqref{eq:dyadiccond}}\\
            &< K_{r,s}(x\mid x)+2C\log r.\tag{by~\eqref{eq:dyadic}}
        \end{align*}
        As $\lfloor x\rfloor_r\in X$,
        \begin{align*}
            K_r(x)&\leq K(\lfloor x\rfloor_r)+C\log r\tag{by~\eqref{eq:dyadic}}\\
            &\leq K(\lfloor x\rfloor_r\mid z)+O(\log r)\tag{by~\eqref{eq:soi}}\\
            &\leq \log|X|+O(\log r).\tag{by~\eqref{eq:cardcomp}}
        \end{align*}
        Also, there are fewer than $2^{K_r(x)+2C\log r+1}$ objects with complexity at most $K_r(x)+2C\log r$, so we can bound $|X|$ on both sides as in~\eqref{eq:X}.
    \end{proof}
    \begin{claim}
        There is some constant $C_1\in\N$ such that
        \begin{equation}\label{eq:Ne_card}
            |N_e|>2^{K_r(x)-\delta r- C_1\log r}.
        \end{equation}
    \end{claim}
    \begin{proof}
        Observe that each set $N_v$ is uniformly computably enumerable relative to $v$, given $z$. Since
        \[|p_e\lfloor x\rfloor_r-p_e x|\leq |\lfloor x\rfloor_r-x|< 2^{1/2-r}\leq 2^{1/2-s},\]
        we have
        \begin{align*}
            K_{r,s}^e(p_e\lfloor x\rfloor_r\mid p_e\lfloor x\rfloor_r)&\leq K_{r-2,s+2}^e(p_e\lfloor x\rfloor_r\mid p_e\lfloor x\rfloor_r)+4C_0\log r\tag{by~\eqref{eq:precision} and~\eqref{eq:condprecision}}\\
            &\leq K_{r,s}^e(p_e x\mid p_e x)+4C_0\log r\\
            &<\frac{K_{r,s}(x\mid x)}{2}-\varepsilon r+4C_0\log r.\tag{by assumption}
        \end{align*}
        As we have already shown $\lfloor x\rfloor_r\in X$, we conclude that $\lfloor x\rfloor_r\in N_e$. Therefore,
        \begin{align*}
            K_r(x)&\leq K^e_r(x)+\delta r\tag{by~\eqref{eq:x}}\\
            &\leq K^e(\lfloor x\rfloor_r)+\delta r+C\log r\tag{by~\eqref{eq:dyadic}}\\
            &\leq K^e(\lfloor x\rfloor_r\mid z)+\delta r +O(\log r)~\tag{by~\eqref{eq:soi}}\\
            &\leq \log|N_e|+2\log\log|N_e|+\delta r+O(\log r)\tag{by~\eqref{eq:cardcomp}}\\
            &\leq \log|N_e|+\delta r+O(\log r),
        \end{align*}
        since $|N_e|\leq |X|\leq 2^{K_r(x)+C\log r+1}=2^{O(\log r)}$. Rearranging and choosing $C_1\in\N$ accordingly yields~\eqref{eq:Ne_card}.
    \end{proof}
    
    \begin{claim}
        \begin{equation}\label{eq:V_card}
            |V|\geq 2^{K_r(e)-O(\log r)}.
        \end{equation}
    \end{claim}
    \begin{proof}
        First note that $S$ is uniformly computably enumerable given $z$, and $\lfloor e\rfloor_r\in S$; it is immediate from~\eqref{eq:dyadic} that $K(\lfloor e\rfloor_r)< K_r(e)+C\log r$, and
        \begin{align*}
            K_t(\lfloor e\rfloor_r)-K_t(e)&\leq \left|K_t(\lfloor e\rfloor_r)-K(\lfloor e\rfloor_t)\right|+\left|K(\lfloor e\rfloor_t)-K_t(e)\right|\\
            &=\left|K_t(\lfloor e\rfloor_r)-K(\lfloor\lfloor e\rfloor_r\rfloor_t)\right|+\left|K(\lfloor e\rfloor_t)-K_t(e)\right|\\
            &<2C\log t.
        \end{align*}
        It follows that $V_1$ is also uniformly computably enumerable given $z$. By~\eqref{eq:Ne_card}, $\lfloor e\rfloor_r\in V_1$. Furthermore, there must be many other members of $V_1$ in the same $t$-dyadic interval as $\lfloor e\rfloor_r$; let
        \[T=\{d\in V_1\mid \lfloor d\rfloor_t=\lfloor e\rfloor_t\},\]
        and observe that $T$ is uniformly computably enumerable given $\lfloor e\rfloor_t$ and $z$. Therefore,
        \begin{align*}
            K_r(e)&\leq K(\lfloor e\rfloor_r)+C\log r\tag{by~\eqref{eq:dyadic}}\\
            &\leq  K(\lfloor e\rfloor_r\mid z)+K(\lfloor e\rfloor_t)+O(\log r)\tag{by~\eqref{eq:soi}}\\
            &\leq \log |T|+K_t(e)+O(\log r)\tag{by~\eqref{eq:cardcomp} and~\eqref{eq:dyadic}}
        \end{align*}
        so there is some constant $C_2\in\N$ such that
        \begin{equation}\label{eq:crowded}
            |T|\geq 2^{K_r(e)-K_t(e)-C_2\log r}.
        \end{equation}
        This is the constant $C_2$ we use in defining $M$.

        By~\eqref{eq:crowded}, $\lfloor e\rfloor_t$ will be labeled crowded. It follows that there must be many crowded $t$-dyadic directions. In particular, let $F$ be the set of directions labeled crowded. This set is uniformly computably enumerable given $z$, so
        \begin{align*}
            K_t(e)&\leq K(\lfloor e\rfloor_t)+C\log t\tag{by~\eqref{eq:dyadic}}\\
            &\leq K(\lfloor e\rfloor_t\mid z)+O(\log r)\tag{by~\eqref{eq:soi}}\\
            &\leq \log|F|+O(\log r).\tag{by~\eqref{eq:cardcomp}}
        \end{align*}
        Thus, at least $2^{K_t(e)-O(\log r)}$ $t$-dyadic directions will be labeled crowded. For each $d\in F$, the machine $M$ prints exactly $\lfloor 2^{K_r(e)-K_t(e)-C_2\log r}\rfloor$ $r$-dyadic directions, so we conclude that $|V_2|\geq 2^{K_r(e)-O(\log r)}$.
        As $|V|=|V_2|$, this proves~\eqref{eq:V_card}.
    \end{proof}

    \begin{claim}
        The sets $X$ and $V$, the collection $\{N_v\mid v\in V\}$, $\alpha=2^{-\delta r-(C+C_2)\log r-1}$,
        and the binary relation $|u-v|<2^{-t}$ satisfy the hypotheses of Lemma~\ref{lem:comb}.
    \end{claim}
    \begin{proof}
        From~\eqref{eq:X} and~\eqref{eq:V1}, we can see that the condition $|N_v|\geq\alpha|X|$ is met by all $v\in V$. For~\eqref{eq:fewsim}, observe that $|u-v|<2^{-t}$ implies $|\lfloor u\rfloor_t-\lfloor v\rfloor_t|\leq 2^{-t}$, so $v$ can only be $2^{-t}$-close to directions in its own $t$-dyadic interval or one of the two adjacent $t$-dyadic intervals. The enumeration of $V_2$ by $M$ ensures that each $t$-dyadic interval's intersection with $V_2$ contains either zero or $\lfloor 2^{K_r(e)-K_t(e)-C_2\log r}\rfloor$ directions, and this property is carried over to $V$ as each $v_d$ is in the same $t$-dyadic interval as $d$. Therefore,
        \begin{align*}
            |\{u\in V\mid |u-v|<2^{-t}\}|&\leq 3\lfloor 2^{K_r(e)-K_t(e)-C_2\log r}\rfloor\\
            &\leq 2^{-K_t(e)+O(\log r)}|V|\\
            &\leq 2^{-t\dim(e)+o(t)+O(\log r)}|V|\\
            &=2^{-\varepsilon r\dim(e)/2+o(r)}|V|\tag{$t=\lfloor \varepsilon r/2\rfloor$}\\
            &<2^{-2\delta r-2(C+C_2)\log r-4}|V|\\
            &=\frac{\alpha^2}{4}|V|
        \end{align*}
        holds when $r$ is sufficiently large, as $\delta<\varepsilon\dim(e)/4$.
    \end{proof}

    Hence, we can apply Lemma~\ref{lem:comb} to let $u,v\in V$ be such that $|u-v|\geq 2^{-t}$ and
    $|N_u\cap N_v|>2^{K_r(x)-2\delta r-O(\log r)}$.
    Since this intersection is large, it contains high-complexity points relative to every oracle; in particular, let $d\in N_u\cap N_v$ such that
    \begin{equation}\label{eq:Kuvd}
        K^{u,v}(d)\geq K_r(x)-2\delta r-O(\log r).
    \end{equation}
    
    Let $\pi_u$ and $\pi_v$ be witnesses to $K^u(\lfloor p_u d\rfloor_r \mid \lfloor p_u  d\rfloor_s)$ and $K^v(\lfloor p_v d\rfloor_r\mid \lfloor p_v d\rfloor_s)$, respectively, and define the sets $P_u=\{U^u(\pi_u,d'):d'\in\D_s\}$ and $P_v=\{U^v(\pi_v,d'):d'\in\D_s\}$. Intuitively, each of these sets contains, within each $s$-dyadic interval, at most one candidate value for an approximation of a projection of $d$. In particular, $\lfloor p_u d\rfloor_r\in P_u$ and $\lfloor p_v d\rfloor_r\in P_v$. Let $D_{u,v}$ be the set of all $r$-dyadic points that share these properties; that is,
    \[D_{u,v}=\left\{y\in\D_r^2\mid \lfloor p_u y\rfloor_r\in P_u\text{ and }\lfloor p_v y\rfloor_r\in P_v\right\}.\]
    For any two points $y,y'\in D_{u,v}$ such that $|y-y'|<2^{1-s}$, we have
    \begin{align*}
        |\lfloor p_u y\rfloor_r -\lfloor p_u y'\rfloor_r|&\leq |\lfloor p_u y\rfloor_r -p_u y|+|p_u y-p_u y'|+|p_u y'-\lfloor p_u y'\rfloor_r|\\
        &<2^{-r}+|y-y'|+2^{-r}\\
        &< 2^{2-s},
    \end{align*}
    and similarly $|\lfloor p_v y\rfloor_r -\lfloor p_v y'\rfloor_r|<2^{2-s}$. As each $s$-dyadic interval contains at most one element of $P_u$ and one element of $P_v$, this implies that the set
    \[P_{u,v}=\{(\lfloor p_u y\rfloor_r,\lfloor p_v y\rfloor_r)\mid y\in D_{u,v}\text{ and }\lfloor y\rfloor_s=\lfloor d\rfloor_s\}\]
    has cardinality $O(1)$. Furthermore, since $u$ and $v$ are $2^{-t}$-separated, the set
    \[\hat{D}_{u,v}=\{y\in\D_r^2\mid \lfloor p_u y\rfloor_r=\lfloor p_u d\rfloor_r\text{ and }\lfloor p_v y\rfloor_r=\lfloor p_v d\rfloor_r\}\]
    is contained within a parallelogram of diameter $O(2^{t-r})$, meaning some point in $\hat{D}_{u,v}$ is a precision-$(r-t-O(1))$ approximation of $d$ that can be computed, relative to $u$ and $v$, given $\pi_u$, $\pi_v$, a witness to $K(\lfloor d\rfloor_s)$, $O(1)$ bits to specify $(\lfloor p_u d\rfloor_r,\lfloor p_v d\rfloor_r)$ within the set $P_{u,v}$, and $O(1)$ bits for the description of the Turing machine that performs this computation. That is,
    \begin{equation}\label{eq:compute_d}
        K^{u,v}_{r-t-O(1)}(d)\leq |\pi_u|+|\pi_v|+K(\lfloor d\rfloor_s)+O(1).
    \end{equation}
    Therefore,
    \begin{align*}
        K_{r,s}(x\mid x)&=K(\lfloor x\rfloor_r\mid \lfloor x\rfloor_s)+O(\log r)\tag{by~\eqref{eq:dyadic} and~\eqref{eq:dyadiccond}}\\
        &= K(\lfloor x\rfloor_r)-K(\lfloor x\rfloor_s)+O(\log r)\tag{by~\eqref{eq:soi}}\\
        &=K_r(x)-K(\lfloor x\rfloor_s)+O(\log r)\tag{by~\eqref{eq:dyadic}}\\
        &\leq K^{u,v}(d)-K_s(x)+2\delta r+O(\log r)\tag{by~\eqref{eq:Kuvd}}\\
        &= K_r^{u,v}(d)-K(\lfloor d\rfloor_s)+2\delta r+O(\log r)\tag{by~\eqref{eq:dyadic}}\\
        &\leq K^{u,v}_{r-t-O(1)}(d)+2t-K(\lfloor d\rfloor_s)+2\delta r+O(\log r)\tag{by~\eqref{eq:precision}}\\
        &\leq |\pi_u|+|\pi_v|+2t+2\delta r+O(\log r)\tag{by~\eqref{eq:compute_d}}\\
        &= K^u(\lfloor p_u d\rfloor_r \mid \lfloor p_u  d\rfloor_s)+K^v(\lfloor p_v d\rfloor_r \mid \lfloor p_v  d\rfloor_s)+2t+2\delta r+O(\log r)\\
        &= K_{r,s}^u(p_u d\mid p_u  d)+K_{r,s}^v(p_v d\mid p_v  d)+2t+2\delta r+O(\log r)\tag{by~\eqref{eq:dyadic} and~\eqref{eq:dyadiccond}}\\
        &\leq K_{r,s}(x\mid x)-2\varepsilon r+2t+2\delta r+O(\log r)\tag{$d\in N_u\cap N_v$}\\
        &<K_{r,s}(x\mid x)-\varepsilon r+2\delta r+O(\log r).\tag{$t=\lfloor \varepsilon r/2\rfloor$}
    \end{align*}
    As $\delta<\varepsilon\dim(e)/4\leq \varepsilon/4$, this implies $\varepsilon r/2< O(\log r)$, a contradiction for sufficiently large $r$. We conclude that~\eqref{eq:orth} holds.
\end{proof}

\section{Hausdorff Dimension Bounds}\label{HDB}

We now prove our Hausdorff dimension bound for pinned distance sets. For all $s\in(0,1)$, this is a strict improvement over the $\frac{s-2+\sqrt{4+s^2}}{2}$ lower bound of Shmerkin and Wang~\cite{ShmWang25} and the $s\left(1-\frac{2-s}{2(1+2s-s^2)}\right)$ lower bound of Fiedler and Stull~\cite{FieStu24a}, which are the best previously known bounds. For example, where those bounds coincide at $s\approx 0.47671$, they give $\sup_x\dimH(\Delta_x E)\gtrsim 0.26637$, while our bound gives $\sup_x\dimH(\Delta_x E)\gtrsim 0.35753$.

\begin{figure}[t]
\centering
\begin{tikzpicture}
\begin{axis}[
  width=0.75\linewidth, height=0.75\linewidth,
  xmin=0, xmax=1, ymin=0, ymax=0.8,
  grid=none,
  xlabel={$\dimH(E)$},
  ylabel={lower bound on $\sup_{x\in E}\dimH(\Delta_x E)$},
  tick label style={/pgf/number format/fixed},
  legend style={at={(0.02,0.98)},anchor=north west,draw=none,fill=none,font=\footnotesize}
]
  \addplot[very thick] {0.75*x};
  \addlegendentry{$\tfrac{3}{4}\dimH(E)$ (this work)}
  \addplot[dashed,thick,domain=0:1,samples=400]
    { x * (1 - (2 - x) / (2 * (1 + 2*x - x^2))) };
  \addlegendentry{Fiedler and Stull~\cite{FieStu24a}}
  \addplot[dotted,thick,domain=0:1,samples=400]
    { 0.5 * (x - 2 + sqrt(4 + x^2)) };
  \addlegendentry{Shmerkin and Wang~\cite{ShmWang25}}
\end{axis}
\end{tikzpicture}
\caption{Pinned distance set lower bounds for $\dimH(E)\in[0,1]$.}
\label{fig:pinned-comparison}
\end{figure}
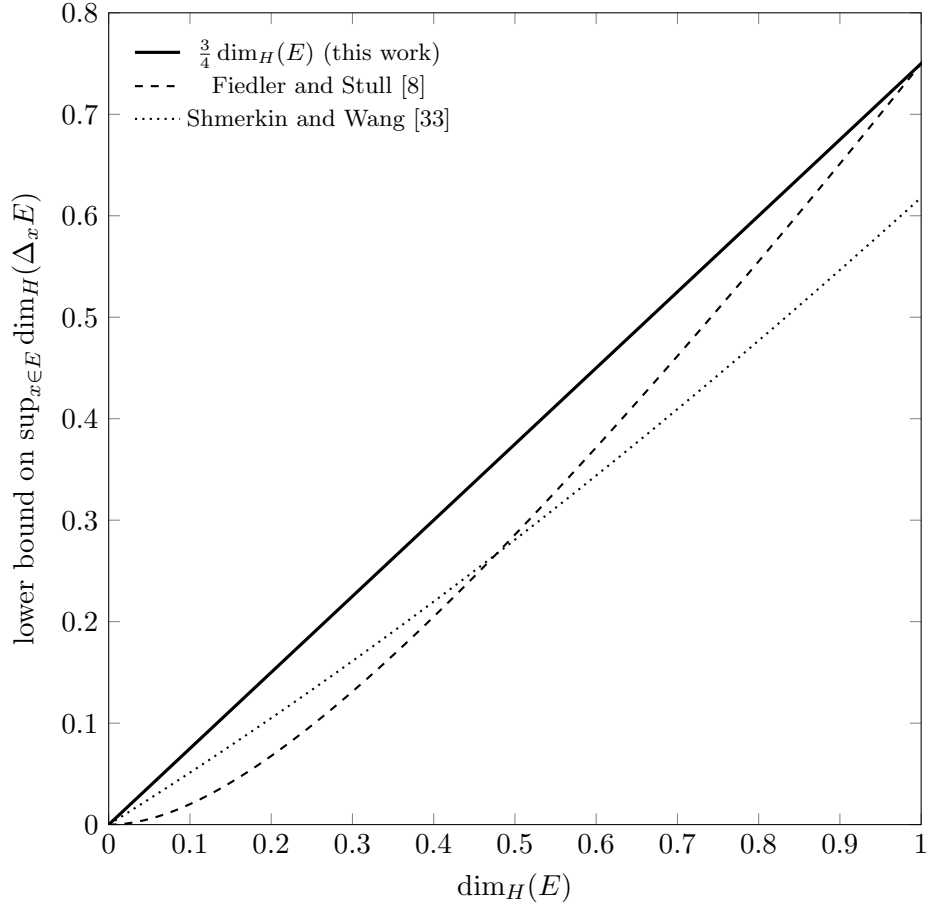

Our proof relies on an effective dimension bound, Theorem~\ref{thm:effDistanceThm}, for which will use the following results of Lutz and Stull~\cite{LutStu24}, Stull~\cite{Stull22c}, and Fiedler and Stull~\cite{FieStu24a}.
    \begin{lemma}[Lutz and Stull~\cite{LutStu24}]\label{lem:oracles}
    Let $z\in\R^2$, $\eta \in\Q_+$, and $r\in\N$. Then there is an oracle $D=D(r,z,\eta)$ with the following properties.
    \begin{itemize}
        \item[(i)] For every $t\leq r$,
        $K^D_t(z)=\min\{\eta r,K_t(z)\}+O(\log r)$.
        \item[(ii)] For every $m,t\in\N$ and $y\in\R^m$,
        $K^{D}_{t,r}(y\mid z)=K_{t,r}(y\mid z)+ O(\log r)$,
        and
        \[K_t^{z,D}(y)=K_t^z(y)+ O(\log r).\]
        \item[(iii)] If $B\subseteq\N$ satisfies $K^B_r(z) \geq K_r(z) - O(\log r)$, then $K_r^{B,D}(z)\geq K_r^D(z) - O(\log r)$.
        \item[(iv)] For every $t\in\N$, $u\in\R^n, w\in\R^m$
        \[K_{r,t}(u\mid w) \leq K^D_{r,t}(u\mid w) + K_r(z) - \eta r + O(\log r)\,.\]
    \end{itemize}
    In particular, this oracle $D$ encodes $\sigma$, the lexicographically first time-minimizing witness to $K(\lfloor z\rfloor_r\mid \lfloor z\rfloor_s)$, where $s = \max\{t \leq r \mid K_{t-1}(z) \leq \eta r\}$.
    \end{lemma}
    
    \begin{observation}[Stull~\cite{Stull22c}]\label{obs:geometricObs}
        Let $x,y, z\in\R^2$ such that $\vert x - y\vert  = \vert x-z\vert $. Let $e_1 = \frac{y - x}{\vert x - y\vert }$ and $e_2 = \frac{y-z}{\vert y-z\vert }$. Let $w\in\R^2$ be the midpoint of the line segment from $y$ to $z$, and $e_3 = \frac{w-x}{\vert x-w\vert }$. Then 
        \begin{enumerate}
            \item $|p_{e_1} y - p_{e_1} z| = \frac{\vert y - z\vert ^2}{2\vert x-y\vert }$, and
            \item $\vert e_1 - e_3\vert  < \frac{\vert y-z\vert }{\vert x-y\vert }$.
        \end{enumerate}
    \end{observation}
    
    \begin{observation}[Stull~\cite{Stull22c}]\label{obs:existProj}
        Let $x\in \R^2$, $d \in \R_{> 0}$, $p \in \Q^2$, and $r \in \N$ such that $\vert d  - \vert x - p\vert  \vert < 2^{-r}$. Then there is a $y \in \R^2$ such that $\vert  p - y\vert  < 2^{-r}$ and $d = \vert x-y\vert $.
    \end{observation}
    
    \begin{lemma}[Stull~\cite{Stull22c}]\label{lem:lowerBoundOtherPointDistance}
    Let $x, y\in \R^2$ and $r\in \N$. Let $z\in \R^2$ such that $\vert x-y\vert = \vert x-z\vert$. Then for every $A\subseteq \N$,
    \begin{equation}
    K^A_r(z) \geq K^A_t(y) + K^A_{r-t, r}(x\mid y) - K_{r-t}(x\mid p_{e^\prime} x, e^\prime) - O(\log r),
    \end{equation}
    where $e^\prime = \frac{y-z}{\vert y-z\vert}$ and $t = -\log \vert y-z\vert$.
    \end{lemma}
    
    \begin{lemma}[Fiedler and Stull~\cite{FieStu24a}]\label{lem:distanceEnumeration}
        Suppose that $B\subseteq\N$, $x, y\in\R^2$, $t<r\in\N$, $\sigma\in(0, 1]$ and $\eta, \varepsilon \in\Q_+$ satisfy the following conditions.
        \begin{itemize}
            \item[(i)] $K^B_r(y)\leq \left(\eta +\frac{\varepsilon}{2}\right)r$.
            \item[(ii)] For every $z \in \mathcal{B}_{2^{-t}}(y)$ such that $\vert x-y\vert = \vert x-z\vert$, \[K^B_{r}(z)\geq \eta r + \min\{\varepsilon r, \sigma (r-s) -\varepsilon r\}\,,\]
            where $s=-\log\vert y-z\vert\leq r$.
        \end{itemize}
        Then,
        \[K_{r,t}^{B, x}(y \mid y) \leq K^{B,x}_{r,t}( \vert x-y\vert\mid y) + \frac{3\varepsilon}{\sigma} r + K(\varepsilon,\eta)+O(\log r)\,.\]
    \end{lemma}
    
    The following lemma is a consequence of a recent result of Cs\"ornyei and Stull \cite{CsoStu25rad}.
    \begin{lemma}\label{lem:slicingLines}
        Let $x\in \R^2$, $e\in \Sphere^1$, $\sigma, \in \Q \cap (0,1)$, $\varepsilon>0$, $C\in\N, A\subseteq\N$, and $t,r\in \N$. Suppose that $r$ is sufficiently large and that the following hold
    \begin{enumerate}
        \item $0<\sigma<\dim^A(x)\leq 1$,
        \item $\frac{r}{C}\leq t <r$,
        \item $K_s^{A, x}(e)\geq\sigma s - C\log s$ for all $s\leq t$
    \end{enumerate}
    Then
    \begin{equation*}
        K_r^A(x\mid p_ex, e)\leq K^A_r(x) - \frac{\sigma}{2}(r+t) + \varepsilon r.
    \end{equation*}
    \end{lemma}
    \begin{proof}
        Let $\ell$ be the line containing $x$ in the direction of $e$. By Theorem 3.5 of \cite{CsoStu25rad}, 
        \begin{equation}
            K^A_r(\ell) \geq K^A_r(e\mid x) + \min\{\dim^A(x), \frac{\frac{\sigma t}{r}+\dim^A(x)}{2}, 1\}r - \varepsilon r - O(\log r).
        \end{equation}
    
        Therefore, we see that 
        \begin{align*}
            K^A_r(x\mid p_e x, e) &= K^A_r(x \mid \ell) + O(\log r)\\
            &= K^A_r(x) + K^A_r(e\mid x) - K^A_r(\ell) + O(\log r)\\
            &\leq K^A_r(x) - \frac{\frac{\sigma t}{r}+\dim^A(x)}{2} r + \varepsilon r + O(\log r)\\
            &\leq K^A_r(x) - \frac{\sigma t + \sigma r}{2} + \varepsilon r\\
            &= K^A_r(x) - \frac{\sigma}{2}(r+t) + \varepsilon r.
        \end{align*}
    \end{proof}

    \begin{theorem}\label{thm:effDistanceThm}
    Suppose that $x, y\in\R^2$, $e = \frac{y-x}{\vert y-x\vert}$, $0 < \sigma < 1$ and $A\subseteq\N$ satisfy the following conditions.
    \begin{description}
        \item[\textup{(C1)}] $\dim^A(x), \dim^A(y) > \sigma$
        \item[\textup{(C2)}] $K^{A,x}_r(e) > \sigma r - O(\log r)$ for all $r$.
        \item[\textup{(C3)}] $K^{A, x}_r(y) \geq K^{A}_r(y) - O(\log r)$ for all sufficiently large $r$. 
        \item[\textup{(C4)}] $K^{A}_{t,r}(e\mid y) > \sigma t - O(\log r)$ for all sufficiently large $r$ and $t \leq r$.
    \end{description} Then $\dim^{x,A}(\vert x-y\vert) \geq  \frac{3\sigma}{4}$.
    \end{theorem}
    \begin{proof}
        Let $\varepsilon \in \Q_+$. Let $r\in\N$ be sufficiently large. Let $t\in\N$ be the precision
        \[
            t = \max\left\{s < r \;\middle\vert\; K^A_{s}(y) \leq \frac{\sigma r}{2}\right\}.
        \]
        By (C1) such a $t$ exists. Moreover, by \eqref{eq:precision},
        \begin{equation}\label{eq:boundComplexityAtT}
        K^{A}_t(y) = \frac{\sigma r}{2}  - O(\log r).
        \end{equation}

        It is immediate from our assumptions that the conditions of Theorem \ref{thm:pinned}, relative to $A$, are satisfied for $x,y, \varepsilon / 2$ and $t$. Therefore, 
        \begin{equation}\label{eq:lowerBoundDistanceAtT}
        K^{A,x}_{t}(\vert x-y\vert ) \geq \frac{\sigma r}{4} - \varepsilon r.
        \end{equation}
        
        We now show that 
        \begin{equation}\label{eq:mainThmEffDimEq1}
        K^{A,x}_{r,t}(\vert x-y\vert \mid y) \geq \frac{\sigma r}{2} - 4\varepsilon r.
        \end{equation}
    
        Let $\eta\in\Q$ be rational such that $\sigma r - 5\varepsilon r\leq \eta r \leq \sigma r- 4\varepsilon r$. Let $D = (r, y, \eta)$ be the oracle of Lemma \ref{lem:oracles} relative to $A$. Let $z \in \R^2$  such that $z\in \mathcal{B}_{2^{-t}}(y)$, and $\vert x-y\vert  = \vert x-z\vert $; this exists by Observation~\ref{obs:existProj}. Let $s = -\log \vert y-z\vert $. There are two cases to consider. 
        
        \paragraph*{Case 1: $s \geq \frac{r}{2} - \log r$.}
            By Observation \ref{obs:geometricObs}, 
            \[
                \vert p_e y - p_e z\vert  < r^2 2^{-r}.
            \]
            Let $r^\prime = r - 2\log r$. Then, by Lemma \ref{lem:lowerBoundOtherPointDistance},
            \begin{align*}
                K^{A,D}_r(z) &\geq K^{A,D}_{r^\prime}(z)\notag\\
                &\geq K^{A,D}_s(y) + K^{A,D}_{r^\prime-s,r^\prime}(e\mid y) - O(\log r)\tag{by Lemma \ref{lem:lowerBoundOtherPointDistance}}\\
                &= K^{A,D}_s(y) + K^{A,D}_{r-s,r}(e\mid y) - O(\log r)\tag{by \eqref{eq:precision}}\\
                &\geq K^{A,D}_s(y) + K^A_{r-s,r}(e\mid y)  - O(\log r)\tag{by Lemma \ref{lem:oracles}(ii)}\\
                &\geq K^{A,D}_s(y) + \sigma(r - s) -\frac{\varepsilon r}{2} - O(\log r)\tag{by (C4)}.
            \end{align*}
            By Lemma \ref{lem:oracles}(i),
            \[
                K^{A,D}_s(y) = \min\{\eta r, K^A_s(y)\} + O(\log r).
            \]
            It is not difficult to verify that, regardless of which term is minimal, 
            \begin{equation*}
                K^{A,D}_r(z) \geq \eta r + \min\{\varepsilon r, r - s -\varepsilon r\}.
            \end{equation*}
            
        \paragraph*{Case 2: $s < \frac{r}{2} - \log r$.}
            Let $e^\prime = \frac{y-z}{\vert y-z\vert }$. It is clear by our assumptions that the conditions of Lemma~\ref{lem:slicingLines} are satisfied for $x, e^\prime, \varepsilon, s, r-s$ and $C^\prime$, where $C^\prime$ depends only on $x, y$ and $e$.
        
            By Lemma \ref{lem:lowerBoundOtherPointDistance},
            \begin{align}
                K^{A,D}_r(z) &\geq K^{A,D}_s(y) + K^{A,D}_{r-s, r}(x\mid y)- K^{A,D}_{r-s}(x\mid p_{e^\prime}x, e^\prime) - O(\log r)\tag{by Lemma \ref{lem:lowerBoundOtherPointDistance}}\\
                &\geq K^{A,D}_s(y) + K^A_{r-s, r}(x\mid y) K^{A,D}_{r-s}(x\mid p_{e^\prime}x, e^\prime) - O(\log r)\tag{by Lemma \ref{lem:oracles}(ii)}\\
                &\geq K^{A,D}_s(y) + K^{A}_{r-s}(x)  - K^{A,D}_{r-s}(x\mid p_{e^\prime}x, e^\prime) - O(\log r)\tag{by (C3)}\\
                &\geq K^{A,D}_s(y) + K^{A}_{r-s}(x)  - K^{A}_{r-s}(x\mid p_{e^\prime}x, e^\prime) - O(\log r)\label{eq:mainThmEffDim2}.
            \end{align}
            By (C2) we can apply Lemma~\ref{lem:slicingLines}, relative to $A$, which yields
            \begin{center}
            $K^{A}_{r-s}(x\mid p_{e^\prime} x, e^\prime) \leq K^A_{r-s}(x) - \frac{\sigma}{2}(r)  + \varepsilon (r-s)$.
            \end{center}
            By combining this with (\ref{eq:mainThmEffDim2}) we see that
            \begin{equation*}
                K^{A,D}_r(z)\geq K^{A,D}_s(y) + \frac{\sigma}{2}r - \varepsilon r - O(\log r).
            \end{equation*}
            Since $s \geq t$, by (\ref{eq:boundComplexityAtT}), $K^{A,D}_s(y) \geq \frac{\sigma r}{2} - O(\log r)$. Thus,
            \begin{align*}
            K^{A,D}_r(z) &\geq K^{A,D}_s(y) + \frac{\sigma}{2}r - \varepsilon r - O(\log r)\\
            &\geq \sigma r   - \varepsilon r - O(\log r)\\
            &\geq (\eta + \varepsilon) r.
            \end{align*}
            
        \bigskip
        Applying Lemma \ref{lem:distanceEnumeration}, relative to $(A,D)$, yields
        \begin{equation*}
            K_{r,t}^{A,D,x}(y \mid y) \leq K^{A,D,x}_{r,t}(\vert x-y\vert \mid y) +3\varepsilon r + K(\varepsilon,\eta)+O(\log r).
        \end{equation*}
        Rearranging and using the fact that additional information does not increase complexity shows that
        \begin{equation}\label{eq:mainThmEffDim4}
            K^{A,x}_{r,t}(\vert x-y\vert \mid y) \geq K^{A,D,x}_{r,t}(y\mid y) -4\varepsilon r.
        \end{equation}
        A standard symmetry of information argument, along with (C3), bounds the right hand side, from which we deduce \eqref{eq:mainThmEffDimEq1}.
        
        Combining this with (\ref{eq:mainThmEffDim4}) and (\ref{eq:lowerBoundDistanceAtT}) we have
        \begin{align*}
            K^{A,x}_r(\vert x-y\vert ) &= K^{A,x}_{t}(\vert x-y\vert ) + K^{A,x}_{r,t}(\vert x-y\vert \mid \vert x-y\vert ) \\
            &\geq K^{A,x}_{t}(\vert x-y\vert ) + K^{A,x}_{r,t}(\vert x-y\vert \mid y) \\
            &\geq \frac{\sigma r}{4} + K^{A,D,x}_{r,t}(y\mid y) -5\varepsilon r\\
            &\geq \frac{\sigma r}{4} + \frac{\sigma r}{2} - 10\varepsilon r - O(\log r)\\
            &= \frac{3\sigma r}{4}  - 10\varepsilon r - O(\log r).
        \end{align*}
        Since this bound holds for all sufficiently large $r$, the conclusion follows.
    \end{proof}

\begin{theorem}\label{thm:pinnedclassical}
    Suppose $E\subseteq\R^2$ is an analytic set such that $s = \dimH(E) \leq 1$. Then,
    \[
        \sup\limits_{x\in E} \dimH(\Delta_x E) \geq \frac{3s}{4}.
    \]
\end{theorem}
\begin{proof}
    Let $\sigma < s$. We will use the important fact of geometric measure theory that Hausdorff dimension is inner regular for analytic sets (see, e.g., \cite{Mattila99}). That is, for every analytic set $G$ and every $\epsilon > 0$, there is a compact subset $K$ such that $\dimH K > \dimH G - \epsilon$. In particular, since $E$ is analytic, there is a compact subset $F\subseteq E$ such that $\dimH(F) > \sigma$. Let $A$ be a Hausdorff oracle for $F$ relative to which $F$ is computably compact. Using Proposition 7 of \cite{Stull22c}, $\Delta_x F$ is computably compact relative to $(A,x)$ for every $x\in\R^2$. Hence, by Hitchcock's point-to-set principle, \eqref{eq:pspCompact}, $(A,x)$ is a Hausdorff oracle for $\Delta_x F$, for every $x\in\R^2$. Therefore, since $\sigma$ can be arbitrarily close to $s$, it suffices to show that
    \[
        \sup\limits_{x\in F} \sup\limits_{y\in F} \dim^{A,x}(|x-y|) \geq \frac{3\sigma}{4}.
    \]

As shown in \cite{FieStu24a}, Lemmas 33 and 34, there are points $x, y\in F$ and $e = \frac{y-x}{|x-y|}$ which satisfy the conditions of Theorem~\ref{thm:effDistanceThm}. We may therefore apply Theorem \ref{thm:effDistanceThm}, to conclude that, for any such $x,y$, we have $\dim^{A,x}(|x-y|)\geq 3\sigma/4$, and the conclusion follows.
\end{proof}

Combined with the point-to-set principle, Theorem~\ref{thm:orth} allows us to generalize, from analytic sets to sets with optimal Hausdorff oracles, Bourgain's theorem on the Hausdorff dimension of \emph{exceptional sets} for orthogonal projections. Marstrand~\cite{Marstrand54} proved that, for every analytic $E\subseteq\R^2$ and almost every $e\in\Sphere^1$, $\dimH(p_e E)=\min\{\dimH(E),1\}$. A direction $e$ is thus ``exceptional'' for $E$ if $\dimH(p_e E)<s$ for some $s\leq \min\{\dimH(E),1\}$; lesser values of $s$ correspond to ``more exceptional'' directions. There is a long and active line of research in bounding the Hausdorff dimension of sets of such directions~\cite{Kaufman68,MR2763000,Mattila19,ren2023,CCLLMS26}. We prove the following theorem.

\begin{theorem}[generalizing Bourgain~\cite{MR2763000}]\label{thm:orthclassical}
    Let $E\subseteq\R^2$, and suppose that optimal Hausdorff oracles for $E$ exist. Then
    \[\dimH\left\{e\in \Sphere^1 \;\middle|\; \dimH(p_e E) < \frac{\dimH E}{2}\right\} = 0.\]
\end{theorem}
\begin{proof}
    Let
    \[F=\left\{e\in \Sphere^1 \;\middle|\; \dimH(p_e E) < \frac{\dimH E}{2}\right\},\]
    and assume toward a contradiction that $\dimH F>\eta$ for some $\eta>0$. Let $A\subseteq\N$ be an optimal Hausdorff oracle for $E$. By~\eqref{eq:psp}, there is some $e\in F$ satisfying $\dim^A(e)>\eta$. Let $B$ be a Hausdorff oracle for $p_e E$.
    
    Let $\varepsilon\in (0,\eta)$ and $\delta=\frac{\varepsilon\eta}{4}$. Since the Turing join of $A$ and $B$ is also an optimal Hausdorff oracle for $E$, we can apply the definition of optimal Hausdorff oracles, with the auxiliary oracle $e$, to say there is some $x\in E$ such that
    \begin{equation}\label{eq:orthclassical1}
        \dim^{A,B}(x)\geq \dim^{A,B,e}(x)\geq \dimH(E)-\delta
    \end{equation}
    and, for almost every $r\in\N$,
    \[K_r^{A,B,e}(x)\geq K_r^{A,B}(x)-\delta r.\]
    Therefore, by the $s=0$ case of Theorem~\ref{thm:orth}, for all sufficiently large $r$ we have
    \[K^{A,B,e}_r(p_e x)\geq \frac{K_r^{A,B}(x)}{2}-\varepsilon r.\]
    Taking limits inferior yields
    \begin{equation}\label{eq:orthclassical2}
        \dim^{A,B,e}(p_e x)\geq \frac{\dim^{A,B}(x)}{2}-\varepsilon.
    \end{equation}
    As $B$ is a Hausdorff oracle for $p_e E$, we also have
    \begin{equation}\label{eq:orthclassical3}
        \dimH(p_e E)\geq \dim^B(p_e x)\geq \dim^{A,B,e}(p_e x).
    \end{equation}
    Combining~\eqref{eq:orthclassical1},~\eqref{eq:orthclassical2}, and~\eqref{eq:orthclassical3} gives us
    \[\dimH(p_e E)\geq \frac{\dimH E}{2}-\frac{\delta}{2}-\varepsilon.\]
    Letting $\varepsilon$ approach 0, this contradicts $e\in F$, so we conclude that $\dimH F=0$.
\end{proof}

\printbibliography
    
\end{document}